\documentclass[11pt]{article}
\usepackage{fullpage}
\usepackage[T1]{fontenc}
\usepackage{lmodern}
\usepackage{amsmath,amssymb,amsthm,mathtools}
\usepackage{microtype}
\usepackage{enumitem}
\usepackage[colorlinks,citecolor=blue,bookmarks=true,linktocpage]{hyperref}
\usepackage{fancyhdr}
\usepackage{needspace}
\setlist[enumerate]{leftmargin=*,itemsep=4pt,topsep=4pt}
\usepackage{tcolorbox}
\usepackage[backend=biber,url=false,maxbibnames=100,isbn=false,sortcites=true,style=trad-plain]{biblatex}
\newtheorem{theorem}{Theorem}
\newtheorem{lemma}[theorem]{Lemma}
\newtheorem{corollary}[theorem]{Corollary}
\newtheorem{proposition}[theorem]{Proposition}
\newcommand{\Tr}{\operatorname{Tr}}
\newcommand{\MD}{\operatorname{MD}}
\newcommand{\Fq}{\mathbb F_q}

\newcommand{\size}[1]{\lVert #1\rVert}
\newcommand{\bigoh}{\mathcal{O}}

\title{An ETH-based quasipolynomial lower bound for Dualization}
\author{
Yasuaki Kobayashi\thanks{Hokkaido University, email: \texttt{koba@ist.hokudai.ac.jp}}
\and Kazuhiro Kurita\thanks{Okayama University, email: \texttt{k-kurita@okayama-u.ac.jp}}
\and Kunihiro Wasa\thanks{Hosei University, email: \texttt{wasa@hosei.ac.jp}}
}

\usepackage[capitalize,poorman]{cleveref}
\crefname{theorem}{Theorem}{Theorems}
\crefname{lemma}{Lemma}{Lemmas}
\crefname{proposition}{Proposition}{Propositions}
\crefname{corollary}{Corollary}{Corollaries}

\AddToHook{env/lemma/begin}{%
  \crefalias{theorem}{lemma}}
\AddToHook{env/proposition/begin}{%
  \crefalias{theorem}{proposition}}
\AddToHook{env/corollary/begin}{%
  \crefalias{theorem}{corollary}}

\begin{document}
\maketitle
\thispagestyle{plain}

\begin{abstract}
Dualizing monotone Boolean functions (or equivalently, enumerating minimal transversals in hypergraphs) is a long-standing problem whose output-polynomial-time solvability remains open. 
While various special cases have been extensively studied, the state-of-the-art algorithm for the general case, due to Fredman and Khachiyan, runs in quasipolynomial time. 
This paper presents a subexponential-time reduction from \textsc{3SAT} to the complement of \textsc{Dual}: Given a 3CNF formula with $n$ variables, the reduction constructs hypergraphs $\mathcal H$ and $\mathcal L$ of total size $2^{\bigoh(n^{2/3}(\log n)^{1/3})}$ such that $\mathcal L \subseteq \Tr(\mathcal H)$ and the formula is satisfiable if and only if $\mathcal L\neq \Tr(\mathcal H)$, where $\Tr(\mathcal H)$ is the set of minimal transversals of $\mathcal H$.
As a consequence of this reduction, assuming the Exponential Time Hypothesis (ETH), neither \textsc{Dual} nor \textsc{Dualization} admits an algorithm running in $N^{o(\sqrt{\log N/\log\log N})}$ time, where $N$ is the input size for \textsc{Dual} and the combined input and output size for \textsc{Dualization}.
In particular, \textsc{Dualization} cannot be solved in output-polynomial time under ETH.
\end{abstract}

\section{Introduction}
\label{sec:introduction}

The problems \textsc{Dual} and \textsc{Dualization} are classically formulated in terms of monotone Boolean functions represented by prime conjunctive normal forms
(CNFs)~\cite{EGM03}. 
For a Boolean function $f$, the dual of $f$ is a Boolean function $f^d$ satisfying $f^d(\mathbf x) = \neg f(\neg \mathbf x)$ for all Boolean vectors $\mathbf x$, where $\neg \mathbf x$ is the component-wise complement of $\mathbf x$.
Given two monotone Boolean functions $f_\varphi$ and $f_\psi$ represented by prime CNF formulas $\varphi$ and $\psi$, the
problem \textsc{Dual} asks whether $f_\psi = f_\varphi^d$.
The problem \textsc{Dualization} asks, given the prime CNF $\varphi$ representing $f_\varphi$, to compute the prime CNF $\psi$ representing $f_\psi = f^d_\varphi$. 
These formulations have the following hypergraph interpretation.
For a hypergraph $\mathcal H = (V, \mathcal E)$, let
\[
  f_{\mathcal H}(\mathbf x)=\bigwedge_{E \in \mathcal E}\bigvee_{v \in E}x_v.
\]
A vertex set $X \subseteq V$ is a \emph{transversal} of $\mathcal H$ if it intersects every hyperedge in $\mathcal H$, and it is \emph{minimal} if no proper subset is a transversal.
We denote by $\Tr(\mathcal H)$ the family of minimal transversals of $\mathcal H$.
A truth assignment is a minimal satisfying assignment of $f_{\mathcal H}$ if and only if the set of variables assigned true is a minimal transversal of $\mathcal H$.
Hence, we have
\[
  f_{\mathcal H}(\mathbf x)=
  \bigvee_{X \in \Tr(\mathcal H)}\bigwedge_{v \in X}x_v,
  \qquad
  f^d_{\mathcal H}(\mathbf x)=
  \bigwedge_{X \in \Tr(\mathcal H)}\bigvee_{v \in X}x_v
\]
for a reduced hypergraph $\mathcal H$.\footnote{A hypergraph is \emph{reduced} if no hyperedge is a proper subset of another.}
Thus, for given reduced hypergraphs $\mathcal H$ and $\mathcal L$, \textsc{Dual} is
equivalent to deciding whether $\mathcal L = \Tr(\mathcal H)$~\cite{EG95,EGM03}.\footnote{Strictly speaking, $\mathcal L$ is a set of hyperedges on $V$. However, we will also refer to it as a hypergraph as long as no confusion arises.}
Similarly, the goal of \textsc{Dualization} is, given a reduced hypergraph $\mathcal H$, to enumerate all minimal transversals of $\mathcal H$.
In this paper, we use this hypergraph formulation of \textsc{Dual} and \textsc{Dualization}.

For \textsc{Dualization}, since the number of minimal transversals in $\mathcal H$ can be exponentially larger than its size $||\mathcal H|| \coloneqq |V| + |\mathcal E|$, the running time is measured in an \emph{output-sensitive manner}.
We say that an enumeration algorithm runs in \emph{output-polynomial time} if it runs in $N^{\bigoh(1)}$ time, where $N$ is the combined size of the input and output.
For general hypergraphs, polynomial-time solvability of \textsc{Dual} is equivalent to output-polynomial-time enumeration of minimal transversals~\cite{BI95,EGM03}. 
The decision problem \textsc{Dual} therefore captures the central completeness question in the transversal enumeration problem \textsc{Dualization}.

These problems have been studied for several decades (see \cite{EiterMG08} for a survey), and their general (output-)polynomial time complexity has remained unsettled~\cite{EG95,Mary26}.
Notably, Fredman and Khachiyan~\cite{FK96} developed a quasipolynomial-time algorithm for \textsc{Dual}, with running time $N^{\bigoh(\log N/\log\log N)}$.
This algorithm currently achieves the best-known upper bound on running time in general.
The gap between this upper bound and polynomial time has motivated extensive research for tractable restrictions. 
Classical results include output-polynomial-time enumeration for hypergraphs of bounded rank, and polynomial-delay enumeration for hypergraphs of bounded degree or bounded incidence treewidth~\cite{EGM03}. 
More recently, Mary~\cite{Mary26} obtained incremental polynomial enumeration for hypergraphs of bounded VC-dimension, bringing many earlier positive cases under a common structural condition.

A closely related line of research concerns the enumeration of minimal dominating sets in a graph. 
The problem of enumerating all minimal dominating sets in a graph $G$, called \textsc{Minimal Dominating Set Enumeration}, is a special case of \textsc{Dualization}, as it is equivalent to \textsc{Dualization} of the hypergraph defined by the closed neighborhoods of vertices in $G$. 
On the other hand, Kant\'e, Limouzy, Mary, and Nourine~\cite{KLMN14} showed that an output-polynomial-time algorithm for minimal dominating sets, even on co-bipartite graphs, would yield an
output-polynomial-time algorithm for \textsc{Dualization}, meaning that these two problems are ``polynomially equivalent'' to each other.
Within restricted graph classes, there are many polynomial-delay or output-polynomial-time enumeration results, such as polynomial-delay enumeration on chordal graphs~\cite{KLMNU16}, output-polynomial-time enumeration on $K_t$-free graphs for every fixed $t$~\cite{BDHPR20}, and polynomial-delay enumeration on chordal bipartite graphs~\cite{CDG25}. 
These advances exhibit substantial tractable cases while preserving the general problem as a common obstacle.

The underlying difficulty of \textsc{Dualization} also provides a basis for conditional lower bounds for other enumeration problems. 
Such hardness results state that an output-polynomial-time enumeration algorithm for a target problem would yield one for \textsc{Dualization}. 
The equivalence for \textsc{Minimal Dominating Set Enumeration}  extends to \textsc{Minimal Total Dominating Set Enumeration} and \textsc{Minimal Connected Dominating Set Enumeration} in split graphs~\cite{KLMN14}. 
Similar hardness results are obtained for \textsc{Minimal Connected Vertex Cover Enumeration} and \textsc{Minimal Capacitated Vertex Cover Enumeration}, already on $2$-degenerate bipartite graphs~\cite{KKMMO25}.

This role of \textsc{Dualization} can be extended to several other graph problems. 
The authors of~\cite{CGKMUW19} gave a reduction from \textsc{Dualization} to the problem of enumerating minimal induced Steiner subgraphs in split graphs.
This proves that an output-polynomial-time algorithm for enumerating minimal induced Steiner subgraphs in split graphs yields an output-polynomial-time enumeration for \textsc{Dualization}.
Conditional lower bounds for other enumeration problems based on \textsc{Dualization} are given for minimal Steiner node multicuts~\cite{KK20}, and minimal group Steiner trees even on stars~\cite{KKW22}. 
Minimal resolving sets in graphs and minimal geodetic sets in split graphs are also equivalent to \textsc{Dualization}~\cite{BDM25}.
More recently, the same hardness was shown for enumerating $s$-$t$ hyperpaths in directed $BF$-hypergraphs~\cite{KM26}.

Hardness results stemming from \textsc{Dualization} appear across various domains well beyond graph problems, such as artificial intelligence~\cite{Kha95,EM02,dCM22}, databases~\cite{BFS22}, data mining~\cite{BuzzegaCKKP25}, and convex geometry~\cite{AGK11,DNV21}.

\paragraph{Results of this paper.}
This paper gives a deterministic subexponential reduction from \textsc{3SAT} to the
complement of \textsc{Dual}. 
More precisely, the following theorem is established.
\begin{theorem}\label{thm:main}
There is a deterministic algorithm that, given a 3CNF formula $\varphi$ with $n$
variables, constructs hypergraphs $\mathcal H$ and $\mathcal L$ such that $\mathcal L \subseteq \Tr(\mathcal H)$ and 
\[
 \varphi\text{ is satisfiable}\ \Longleftrightarrow\ \mathcal L \neq \Tr(\mathcal H).
\]
Moreover, the sizes of $\mathcal L$ and $\mathcal H$ are bounded by $2^{\bigoh(n^{2/3}(\log n)^{1/3})}$, and they are  constructed in time $2^{\bigoh(n^{2/3}(\log n)^{1/3})}$.
\end{theorem}

As a consequence of \cref{thm:main}, assuming the Exponential Time Hypothesis (ETH), which
excludes $2^{o(n)}n^{O(1)}$-time algorithms for \textsc{3SAT}~\cite{IPZ01,ImpagliazzoP01}, this construction rules out an algorithm for \textsc{Dual} with running time
\[
  N^{o\!\left(\sqrt{\log N/\log\log N}\right)},
\]
where $N = ||\mathcal H|| + ||\mathcal L||$.
It follows that, under ETH, \textsc{Dual} has no polynomial-time algorithm, and \textsc{Dualization} has no output-polynomial-time enumeration algorithm.
The same assumption also rules out output-polynomial enumeration for the aforementioned enumeration problems.

\section{Proof of \cref{thm:main}}\label{sec:construction}

To prove \cref{thm:main}, we first construct a hypergraph $\mathcal H'$ from a 3CNF formula $\varphi$ such that $\varphi$ is satisfiable if and only if $\mathcal H'$ has a transversal $X$ with a certain condition.
This transversal $X$ is a minimal transversal of $\mathcal H'$ but the converse does not hold.

\subsection{Constructing a hypergraph from a 3CNF formula}
Let $\varphi$ be a 3CNF formula with $n$ variables.
Let $C$ be the set of clauses in $\varphi$.
We can assume that $|C| = \bigoh(n^3)$.
Choose a positive integer $b$ with $b \le n$.
Let $k = \lceil n/b \rceil$ and partition the variables into $k$ nonempty balanced blocks $B_1,\ldots,B_k$ such that each block contains at most $b$ variables.
For $1 \le i \le k$, define pairwise disjoint vertex sets
\[
 W_i =\{(i,\beta) : \beta:B_i \to \{0,1\}\}.
\]
Each vertex $(i, \beta)$ represents a partial truth assignment for variables in $B_i$.
Let $V = \bigcup_{1 \le i \le k} W_i$.
A vertex set $X$ is called a \emph{full selector} if it contains exactly one element of each $W_j$ ($1 \le j \le k$).
Observe that for every truth assignment $\alpha$ of $\varphi$, a full selector is defined as
\[
 X_\alpha=\{(i,\alpha|_{B_i}) : 1\le i \le k\},
\]
where $\alpha|_{B_i}$ is the restriction of $\alpha$ to $B_i$, and every full selector $X$ uniquely determines a truth assignment $\alpha$ with $X_\alpha = X$.
Since $|W_i| \le 2^b$, we have $|V| \le k2^{b}$.

We first construct a hypergraph $\mathcal H' = (V, \mathcal E')$ and then extend it to the hypergraph stated in \cref{thm:main}.
The set of hyperedges in $\mathcal H'$ is defined as $\mathcal E' = \{W_i : 1 \le i \le k\} \cup \{E_c : c \in C\}$, where $E_c$ is defined as
\[ 
    E_c = \{(i, \beta) : \text{some literal of } c \text{ in } B_i \text{ is true under partial assignment } \beta\}.
\]

A full selector is \emph{valid} if it is a transversal of $\mathcal H'$.
The following proposition is easy to derive from the definition.
\begin{proposition}\label{prop:sat-selector}
    A truth assignment $\alpha$ is a satisfying assignment of $\varphi$ if and only if $X_\alpha$ is a valid full selector.
\end{proposition}

If $\alpha$ is a satisfying assignment of $\varphi$, the corresponding vertex set $X_\alpha$ is a minimal transversal of $\mathcal H'$ as every hyperedge $W_i$ is hit by exactly one vertex in $X_\alpha$.
However, the converse does not hold: a minimal transversal may contain more than one vertex in $W_i$. 
The following padding argument prevents this undesired situation, making all such minimal transversals small enough to enumerate within the claimed time bound.

\subsection{Padding hyperedges}\label{sec:padding}

Let $\mathcal H' = (V, E')$ be the hypergraph constructed in the previous subsection, and let $\mathcal E' = \{E'_1, \dots, E'_M\}$.
Fix a prime $q$ and an integer $t \ge 1$ satisfying
\[
 M\le q,\qquad \max_i|W_i|\le q^t,\qquad 2t<k.
\]
We will discuss the choice of these parameters later.

For each $1 \le i \le k$, we label the vertices $u \in W_i$ injectively by polynomials $p_u \in \Fq[x]$ of degree less than $t$. 
Two polynomials $p_u$ and $p_v$ for two vertices $u \in W_i$ and $v \in W_j$ with $i\neq j$ can be identical.
Since $|W_i| \le q^t$, such an injective labeling always exists.
Fix distinct points $\pi_1, \ldots, \pi_M \in \Fq$.
For each $u \in V$ and $1 \le j \le M$, define $h_j(u) = p_u(\pi_j)$.
These functions $h_j$ assign an $M$-dimensional vector $\mathbf h(u) \coloneqq (h_1(u), \dots, h_M(u))$ to each vertex $u \in V$.

\begin{lemma}\label{lem:colors}
For distinct $u,v \in W_i$, we have
\[
 |\{j : h_j(u) = h_j(v)\}|\le t-1.
\]
\end{lemma}

\begin{proof}
As $p_u \neq p_v$, the difference $p_u-p_v$ is a nonzero polynomial of degree at most $t-1$.
It is well known that every nonzero polynomial of degree at most $d$ has at most $d$ roots~\cite{LidlN97}.
Since $h_j(u) = h_j(v)$ if and only if $\pi_j$ is a root of $p_u - p_v = 0$, the lemma follows.
\end{proof}

Now, we replace each edge $E'_j \in \mathcal E'$ with hyperedges by padding some vertices. 
For $1 \le j \le M$ and a vector $\mathbf{a}=(a_1, \ldots, a_k)\in\Fq^k$, define
\[
 P_{j,\mathbf{a}} =\bigcup_{1 \le i \le k} \{u \in W_i : h_j (u) \ne a_i\},
\]
and $E_{j, \mathbf{a}} = E'_j \cup P_{j,\mathbf{a}}$. 
Let $\mathcal E_j = \{E_{j, \mathbf{a}} : \mathbf{a} \in \Fq^k\}$ and let $\mathcal E = \bigcup_{1 \le j \le M} \mathcal E_j$.
Note that each edge in $\mathcal E_j$ contains $E'_j$ as a subset and $|\mathcal E| \le Mq^k$.
Let $\mathcal H = (V, \mathcal E)$.

The padding is designed to preserve \cref{prop:sat-selector} on full selectors while ensuring that every minimal transversal containing two vertices from some $W_i$ is small.
For an edge $E'_j$, the padding $P_{j, \mathbf a}$ for $\mathbf a = (a_1, \dots, a_k)$ excludes from each $W_i$ precisely those vertices whose vectors have $j$-th coordinate $a_i$.
For a full selector $X = \{u_1, \ldots, u_k\}$ with $u_i \in W_i$, by choosing $a_i = h_j(u_i)$, $X$ misses $P_{j,\mathbf a}$.
Thus, if $X$ misses $E'_j$, it also misses the padded edge $E_{j,\mathbf a}$.
On the other hand, if vectors $\mathbf h(u)$ and $\mathbf h(v)$ of two vertices $u, v \in W_i$ differ in some coordinate~$j$, $u$ and $v$ cannot both be excluded from $P_{j, \mathbf a}$ and therefore $\{u, v\}$ hits every edge in $\mathcal E_j$.
To make such a pair hit almost all families, we use polynomial evaluations to ensure that $\mathbf h(u)$ and $\mathbf h(v)$ agree in at most $t-1$ coordinates.
For a transversal $X$ containing such a pair, the following lemma ensures that there is a vertex subset in $X$ that hits all hyperedges in $\mathcal E_j$.

\begin{lemma}\label{lem:group}
For $1 \le j \le M$, a vertex set $X \subseteq V$ hits every hyperedge in $\mathcal E_j$ if and only if $X \cap E'_j \neq \emptyset$ or there are two vertices $u, v \in X \cap W_i$ for some $i$ such that $h_j(u) \neq h_j(v)$.
In the latter case, $\{u, v\}$ hits every hyperedge in $\mathcal E_j$.
\end{lemma}

\begin{proof}
Suppose that either condition holds.
If $X \cap E'_j \neq \emptyset$, we can clearly conclude that $X$ hits every hyperedge in $\mathcal E_j$ as they have $E'_j$ as a subset.
If there are two vertices $u, v \in X \cap W_i$ for some $i$ such that $h_j(u) \neq h_j(v)$, 
at least one of $h_j(u) \neq a_i$ or $h_j(v) \neq a_i$ holds for every $\mathbf{a} = (a_1, \dots, a_k) \in \Fq^k$.
Thus, $\{u, v\}$ hits $P_{j, \mathbf{a}}$ for every $\mathbf{a}$.

Conversely, suppose that $X$ has no vertices in $E'_j$ and $h_j(u) = h_j(v)$ for vertices $u, v \in X \cap W_i$ for each $1 \le i \le k$.
Choose the value $a_i$ to be the common value of $h_j(u)$ for $u \in X \cap W_i$ if $X \cap W_i \neq \emptyset$; choose $a_i$ arbitrarily otherwise.
Since a vertex $u \in W_i$ hits $P_{j, \mathbf{a}}$ if and only if $h_j(u) \neq a_i$, $X \cap E'_j = \emptyset$ implies that $X$ does not hit $E_{j, \mathbf{a}}$.
\end{proof}

\begin{lemma}\label{lem:selectors}
Suppose that a vertex set $X$ contains at most one vertex from each $W_i$.
Then, $X$ is a transversal of $\mathcal H$ if and only if it is a valid full selector.
Moreover, every valid full selector is a minimal transversal of $\mathcal H$.
\end{lemma}

\begin{proof}
Suppose that $X$ is a transversal of $\mathcal H$.
Since $X$ contains at most one vertex from each $W_i$, by~\cref{lem:group}, it hits $E'_j$ for all $E'_j \in \mathcal E'$.
In particular, $X$ hits $W_i$ for $1 \le i \le k$, which means that it is a valid full selector.

Conversely, suppose that $X$ is a valid full selector.
This immediately implies that $X$ hits $W_i$ for all $1 \le i \le k$.
Moreover, by \cref{prop:sat-selector}, the truth assignment $\alpha$ obtained from $X$ is a satisfying assignment of $\varphi$.
Thus, $X$ hits $E_c$ for all $c \in C$.
Hence, $X$ is a transversal of $\mathcal H'$.
As each hyperedge in $\mathcal H$ contains a hyperedge of $\mathcal H'$, $X$ is a transversal of $\mathcal H$ as well.

Since every proper subset of a valid full selector misses some hyperedge $W_i$, every valid full selector $X$ is a minimal transversal of $\mathcal H'$.
By~\cref{lem:group}, it is also a minimal transversal of $\mathcal H$, as $X$ has at most one vertex from each $W_i$.
\end{proof}

\begin{lemma}\label{lem:small}
If a transversal $X$ of $\mathcal H$ contains two vertices from $W_i$ for some $i$, it
contains a transversal of size at most $2t$. 
\end{lemma}

\begin{proof}
Let $u, v \in X \cap W_j$ be distinct and let $I = \{j : h_j(u) = h_j(v)\}$.
By~\cref{lem:colors}, we have $|I| \le t-1$.
For $j \notin I$, we have $h_j(u) \neq h_j(v)$. 
By~\cref{lem:group}, $\{u, v\}$ hits every hyperedge in $\mathcal E_j$ for $j \notin I$.
Let $j \in I$.
Since $X$ is a transversal of $\mathcal H$, by \cref{lem:group}, there is a set $S_j \subseteq X$ of at most two vertices that intersect each edge in $\mathcal E_j$.
Hence, $\{u, v\} \cup \bigcup_{j \in I} S_j \subseteq X$ is a transversal of $\mathcal H$ with size at most $2 + 2|I| \le 2t$.
\end{proof}

\begin{lemma}\label{lemma:decomposition}
Let $\mathcal L = \{X \in \Tr(\mathcal H) : |X| \le 2t\}$.
Then, $\Tr(\mathcal H) \setminus \mathcal L$ is nonempty if and only if $\varphi$ is satisfiable.
\end{lemma}
\begin{proof}
Suppose that there is a minimal transversal $X$ of $\mathcal H$ that does not belong to $\mathcal L$.
Then, we have $|X| > 2t$.
By~\cref{lem:small} and the minimality of $X$, it contains at most one vertex from each $W_i$.
By~\cref{lem:group}, it is a valid full selector, and hence the corresponding truth assignment satisfies $\varphi$ due to \cref{prop:sat-selector}.

Conversely, let $\alpha$ be a satisfying assignment of $\varphi$.
Then, $X_\alpha$ is a valid full selector due to \cref{prop:sat-selector}.
As $X_\alpha$ contains exactly one vertex from each $W_i$, by~\cref{lem:selectors}, $X_\alpha$ is a minimal transversal of $\mathcal H$.
Moreover, as $|X_\alpha| = k > 2t$, it holds that $X_\alpha \in \Tr(\mathcal H) \setminus \mathcal L$.
\end{proof}

From a 3CNF formula $\varphi$, we can construct the hypergraph $\mathcal H$ as follows.
The vertex set $V$ of $\mathcal H$ has at most $k2^b$ vertices and can be computed in time $2^bn^{\bigoh(1)}$ by enumerating all partial assignments of $B_i$ for all $i$.
The set $\mathcal E'$ of hyperedges in $\mathcal H'$ is defined as $\mathcal E' = \{W_i : 1 \le i \le k\} \cup \{E_c : c \in C\}$.
Each clause hyperedge $E_c$ can be constructed by enumerating all the vertices (partial assignments of variable blocks) and checking whether they satisfy the clause $c$.
Thus, the hyperedges in $\mathcal E'$ can be enumerated in time $\bigoh(|V| \cdot |\mathcal E'|) = |V|\cdot n^{\bigoh(1)}$ as $M = |\mathcal E'| = \bigoh(n^3)$.
For padding hyperedges, we construct functions $h_j$ for each $1 \le j \le M$.
This can be done in time $|V|\cdot (n + q)^{\bigoh(1)}$ by enumerating the coefficients of each polynomial $p_u$ and evaluating them at fixed points $\pi_1, \dots, \pi_M$.
For $1 \le j \le M$, we can construct $\mathcal E_j$ in time $|V| \cdot q^k n^{\bigoh(1)}$, and hence the hypergraph $\mathcal H$ is constructed in total time $|V| \cdot q^{k+ \bigoh(1)} \cdot n^{\bigoh(1)}$.
For constructing $\mathcal L$, we first enumerate all the vertex sets of size at most $2t$ and check whether each set belongs to $\Tr(\mathcal H)$.
As $\mathcal H$ and $\mathcal L$ have at most $Mq^k$ and $|V|^{2t}$ hyperedges, respectively, the construction can be done in time $|V|^{2t+\bigoh(1)}\cdot |\mathcal E| = |V|^{2t+\bigoh(1)} \cdot q^kn^{\bigoh(1)}$.
Therefore, from $\varphi$, we can construct hypergraphs $\mathcal H$ and $\mathcal L$ with
\[
    ||\mathcal H|| \le k2^b + Mq^k, \qquad ||\mathcal L|| \le (k2^b)^{2t+1}
\]
in time $(k2^b)^{2t + \bigoh(1)} \cdot q^{k + \bigoh(1)} \cdot n^{\bigoh(1)}$.

\subsection{Setting parameters}\label{sec:size}

To complete the proof of \cref{thm:main}, we need to appropriately choose integer parameters $b, k, q, t$.
Recall that $n$ is the number of variables in $\varphi$.
We can assume that $n \ge 32$.
We set the parameters
\[
 \begin{gathered}
 \ell = \lceil \log_2(n+1) \rceil,\qquad
 t = \max\{z \in \mathbb Z_{\ge 1} : 4z^3 \ell \le n\},\qquad 
 b = t\ell, \qquad 
 k = \lceil n/b \rceil.
 \end{gathered}
\]
The positive integer $t$ is well defined since $4\ell \le n$ holds when $n \ge 32$.
Let $Q = \max\{M, 2^\ell\}$. 
Choose a prime $q$ with $Q \le q < 2Q$.
Such a prime always exists for $Q > 1$.
Since $M = \bigoh(n^3)$ and $2^\ell = \bigoh(n)$, we can find such a prime $q$ in polynomial time in $n$.
In particular, $q = n^{\bigoh(1)}$.

The three requirements used in the previous subsection directly follow from
\[
 M \le Q \le q,\qquad
 |W_j| \le 2^{t\ell} \le q^t,\qquad
 2t < 4t^2 \le \frac{n}{t\ell} \le k.
\]
As $t$ is taken as the maximum $z$ satisfying $4z^3\ell \le n$, we have
\[
    t \le \left(\frac{n}{4\ell}\right)^{1/3} \le t + 1,
\]
and hence $t = \Theta((n / \log n)^{1/3})$, $b = \Theta(n^{1/3}(\log n)^{2/3})$, and $k = \Theta((n/\log n)^{2/3})$.
Therefore, it holds that
\[
    ||\mathcal H|| \le k2^b + Mq^k = 2^{\bigoh(n^{1/3}(\log n)^{2/3})} + n^{\bigoh((n / \log n)^{2/3})} = 2^{\bigoh(n^{2/3}(\log n)^{1/3})} 
\]
and
\[
    ||\mathcal L|| \le (k2^b)^{2t+1} = 2^{\bigoh(n^{2/3}(\log n)^{1/3})}.
\]
The running time for constructing these hypergraphs is bounded by
\[
    (k2^b)^{2t + \bigoh(1)} \cdot q^{k + \bigoh(1)} \cdot n^{\bigoh(1)} = 2^{\bigoh(n^{2/3}(\log n)^{1/3})}.
\]
This completes the proof of \cref{thm:main}.

We remark that the output hypergraph $\mathcal H$ may not be reduced.
Since we can easily remove all hyperedges that properly contain another hyperedge in $||\mathcal H||^{\bigoh(1)}$ time, \cref{thm:main} also holds for reduced hypergraphs.

\subsection{Consequences of \cref{thm:main}}

The reduction presented in \cref{thm:main} proves the following corollary.
\begin{corollary}\label{cor:duality}
Assuming ETH, \textsc{Dual} has no algorithm with running time
\[
 N^{o(\sqrt{\log N/\log\log N})},
\]
where $N = ||\mathcal H|| + ||\mathcal L||$.
In particular, \textsc{Dual} has no polynomial-time algorithm unless ETH fails.
\end{corollary}
\begin{proof}
Let $\mathcal H$ and $\mathcal L$ be the hypergraphs constructed from a 3CNF formula with $n$ variables in the proof of \cref{thm:main}.
Let $s = n^{2/3}(\log n)^{1/3}$. 
Then, we have
\[
    \log N = \bigoh(n^{2/3}(\log n)^{1/3}) = \bigoh(s),
\]
and their construction can be done in $2^{o(n)}$ time.
Since a function $f(x) \coloneqq x^{3/2}/\sqrt{\log x}$ is monotonically increasing for large~$x$,
\[
 \frac{(\log N)^{3/2}}{\sqrt{\log\log N}} = \bigoh\left(\frac{s^{3/2}}{\sqrt{\log s}}\right) = \bigoh\left(\frac{n \sqrt{\log n}}{\sqrt{\log n + \log\log n}}\right) = \bigoh(n).
\]
Hence, an algorithm for \textsc{Dual} with running time $N^{o(\sqrt{\log N/\log\log N})}$ would yield an algorithm for \textsc{3SAT} with running time $2^{o(n)}$. 
\end{proof}

The following well-known reduction gives an ETH-based quasipolynomial time lower bound for \textsc{Dualization}.

\begin{proposition}[\cite{BI95,EGM03}]\label{prop:dualization}
If there is an algorithm $A$ for \textsc{Dualization} running in time $T(N)$ with $N = ||\mathcal H|| + ||\Tr(H)||$, \textsc{Dual} can be solved in $\bigoh(T(M) + M^{\bigoh(1)})$ time with $M = ||\mathcal H|| + ||\mathcal L||$.
\end{proposition}

\begin{proof}
We simulate algorithm $A$ in $T(M)$ steps, rejecting immediately when it generates more than $|\mathcal L|$ sets.
If it terminates without such a rejection, we have $\Tr(\mathcal H)$ and directly compare $\mathcal L$ and $\Tr(\mathcal H)$ in $M^{\bigoh(1)}$ time.
Suppose otherwise.
If $\mathcal L = \Tr(\mathcal H)$, then $N = M$, and hence $A$ terminates within $T(M)$ steps.
Thus, we conclude that $\mathcal L \neq \Tr(\mathcal H)$.
\end{proof}

\begin{corollary}\label{cor:dualization}
Assuming ETH, \textsc{Dualization} has no algorithm with running time
\[
 N^{o(\sqrt{\log N/\log\log N})},
\]
where $N = ||\mathcal H|| + ||\Tr(\mathcal H)||$.
In particular, \textsc{Dualization} has no output-polynomial-time algorithm unless ETH fails.
\end{corollary}

\Cref{cor:dualization} transfers to every enumeration problem that is at least as hard as
\textsc{Dualization} with respect to output-polynomial time.
We first state the transfer for minimal dominating sets, where the reduction of
Kant\'e, Limouzy, Mary, and Nourine~\cite{KLMN14} is explicit enough to preserve the quantitative bound.
For a graph $G$, let $\MD(G)$ denote the family of minimal dominating sets of $G$.

\begin{corollary}\label{cor:mds}
Assuming ETH, \textsc{Minimal Dominating Set Enumeration} has no algorithm with running time
\[
 N^{o(\sqrt{\log N/\log\log N})},
\]
where $N = \size{G} + \size{\MD(G)}$ is the combined size of the input graph $G$ and the output,
even when $G$ is restricted to co-bipartite graphs.
\end{corollary}

More generally, \cref{cor:dualization} applies to every enumeration problem for which an
output-polynomial-time algorithm would yield one for \textsc{Dualization}.

\begin{corollary}\label{cor:transfer}
Assuming ETH, none of the following problems admits an output-polynomial-time algorithm:
\textsc{Minimal Dominating Set Enumeration} on co-bipartite graphs, 
\textsc{Minimal Total Dominating Set Enumeration}, and
\textsc{Minimal Connected Dominating Set Enumeration} on split graphs~\cite{KLMN14};
\textsc{Minimal Connected Vertex Cover Enumeration} and 
\textsc{Minimal Capacitated Vertex Cover Enumeration} on $2$-degenerate bipartite graphs~\cite{KKMMO25};
enumeration of minimal induced Steiner subgraphs on split graphs~\cite{CGKMUW19};
enumeration of minimal Steiner node multicuts~\cite{KK20};
enumeration of minimal group Steiner trees on stars~\cite{KKW22};
enumeration of minimal resolving sets, and of minimal geodetic sets on split graphs~\cite{BDM25};
and enumeration of $s$-$t$ hyperpaths in directed $BF$-hypergraphs~\cite{KM26}.
\end{corollary}

\section{Discussion}
The results in this paper rule out (output-)polynomial-time algorithms for \textsc{Dual} and \textsc{Dualization} under the Exponential Time Hypothesis (ETH) and transfer this obstruction to many related enumeration problems. This paper leaves several questions about the necessary complexity assumption, the boundary of tractable graph classes, and the relationships among dualization-hard problems.

The principal open question would be whether the qualitative conclusion can be obtained under a weaker assumption, such as $\mathrm{P} \neq \mathrm{NP}$.
The reduction used in \cref{thm:main} is indeed insufficient for this implication.
The quasipolynomial-time upper bound of Fredman and Khachiyan~\cite{FK96} also indicates that a conventional hardness proof would be particularly challenging. 
Indeed, coNP-hardness of \textsc{Dual} under polynomial-time many-one reductions would imply $\mathrm{NP}\subseteq\mathrm{QP}$, where $\mathrm{QP}$ denotes deterministic quasipolynomial time. This consequence is not excluded by $\mathrm{P}\neq\mathrm{NP}$ alone, although it would contradict ETH.
Establishing the desired implication may therefore require an approach that does not rely on proving coNP-completeness.

The lower bound also gives a complexity-theoretic reason to pursue structural restrictions on \textsc{Minimal Dominating Set Enumeration}. 
Positive results already cover substantially different and wide classes: Minimal dominating sets can be
enumerated with polynomial delay in chordal graphs~\cite{KLMNU16} and in output-polynomial time in $K_t$-free graphs for every fixed $t$, including all bipartite graphs~\cite{BDHPR20}.
Moreover, as the closed neighborhood hypergraph of a weakly $\gamma$-closed graph has bounded VC-dimension~\cite{LokshtanovS21} for every fixed $\gamma$, the algorithm of Mary~\cite{Mary26} yields an incremental-polynomial bound for \textsc{Minimal Dominating Set Enumeration} on weakly $\gamma$-closed graphs. 
Nevertheless, output-polynomial-time enumeration remains open for comparability graphs, graphs with no induced $C_4$, and graphs with no induced $2K_3$, where $2K_3$ denotes the disjoint union of two triangles~\cite{BDMN26}.

These classes suggest concrete extensions of existing algorithms.
For comparability graphs, incremental-polynomial-time enumeration is known when the underlying poset has bounded dimension~\cite{BDMN26}; removing this restriction would also extend the positive result for bipartite graphs.
For $C_4$-free graphs, the question is whether tractability extends beyond chordal graphs when longer induced cycles are allowed. 
For $2K_3$-free graphs, the known algorithm for $2K_2$-free graphs provides another natural point of
comparison~\cite{BDHPR20}. 
The bounded-VC-dimension algorithm of Mary~\cite{Mary26} gives a further general source of tractability through closed-neighborhood hypergraphs, but does not seem to settle these three classes.
In the opposite direction, co-bipartite graphs already capture the full difficulty of general \textsc{Dualization}~\cite{KLMN14}. 
Consequently, the presented results rule out output-polynomial-time enumeration under ETH for every graph class
containing all co-bipartite graphs, including perfect graphs.

There is also a separate question about strengthening the running time bound within tractable classes. 
For example, bipartite graphs already admit an output-polynomial-time algorithm, while polynomial delay remains open; polynomial delay has recently been obtained for chordal bipartite graphs~\cite{CDG25}.

The reductions discussed in the introduction transfer the exclusion of output-polynomial-time algorithms under ETH to a broad range of enumeration problems.
A further goal is to determine which of these problems also admit a converse reduction to \textsc{Dualization} that preserves output-polynomial time. 
Such equivalences are known for minimal dominating sets~\cite{KLMN14}, minimal resolving sets in general graphs, and minimal geodetic sets in split graphs~\cite{BDM25}. 
Identifying further equivalences requires control over the size of both the transformed instance and its output.

\section{AI disclosure}
The proof of \cref{thm:main} was originally discovered by GPT-6 Astra (OpenAI). The authors subsequently reconstructed and revised the argument and rewrote its exposition. 
The same AI model was also used to prepare drafts of several parts of the manuscript. 
The authors have verified the correctness of the results and take full responsibility for all mathematical claims and the final manuscript.

\section*{Acknowledgement}
This work was supported by JSPS KAKENHI Grant Numbers
JP22K12285, 
JP23K28034, 
JP24K15204, 
JP24H00686, 
JP24H00697, 
JP25K00136, 
JP25K03080, 
JP25K21273. 

\printbibliography

@book{LidlN97,
  author    = {Lidl, Rudolf and Niederreiter, Harald},
  title     = {Finite Fields},
  edition   = {2nd},
  publisher = {Cambridge University Press},
  series    = {Encyclopedia of Mathematics and its Applications},
  volume    = {20},
  year      = {1997},
  isbn      = {978-0-521-39231-0},
  doi       = {10.1017/CBO9780511525926}
}

@article{BI95,
  author = {Bioch, J. C. and Ibaraki, Toshihide},
  title = {Complexity of Identification and Dualization of Positive {Boolean} Functions},
  journal = {Information and Computation},
  volume = {123},
  number = {1},
  pages = {50--63},
  year = {1995},
  doi = {10.1006/inco.1995.1157},
  url = {https://repub.eur.nl/pub/55247}
}

@article{EG95,
  author = {Eiter, Thomas and Gottlob, Georg},
  title = {Identifying the Minimal Transversals of a Hypergraph and Related Problems},
  journal = {SIAM Journal on Computing},
  volume = {24},
  number = {6},
  pages = {1278--1304},
  year = {1995},
  doi = {10.1137/S0097539793250299},
  url = {https://epubs.siam.org/doi/10.1137/S0097539793250299}
}

@article{FK96,
  author = {Fredman, Michael L. and Khachiyan, Leonid},
  title = {On the Complexity of Dualization of Monotone Disjunctive Normal Forms},
  journal = {Journal of Algorithms},
  volume = {21},
  number = {3},
  pages = {618--628},
  year = {1996},
  doi = {10.1006/jagm.1996.0062}
}

@article{EGM03,
  author = {Eiter, Thomas and Gottlob, Georg and Makino, Kazuhisa},
  title = {New Results on Monotone Dualization and Generating Hypergraph Transversals},
  journal = {SIAM Journal on Computing},
  volume = {32},
  number = {2},
  pages = {514--537},
  year = {2003},
  doi = {10.1137/S009753970240639X},
  eprint = {cs/0204009},
  archivePrefix = {arXiv},
  url = {https://arxiv.org/abs/cs/0204009}
}

@misc{Mary26,
  author = {Mary, Arnaud},
  title = {Enumeration of Minimal Transversals of Hypergraphs of Bounded {VC}-Dimension},
  year = {2026},
  note = {Version 4, revised 30 January 2026; first posted in 2024},
  eprint = {2407.00694},
  archivePrefix = {arXiv},
  primaryClass = {math.CO},
  doi = {10.48550/arXiv.2407.00694},
  url = {https://arxiv.org/abs/2407.00694v4}
}

@article{KLMN14,
  author = {Kant{\'e}, Mamadou Moustapha and Limouzy, Vincent and Mary, Arnaud and Nourine, Lhouari},
  title = {On the Enumeration of Minimal Dominating Sets and Related Notions},
  journal = {SIAM Journal on Discrete Mathematics},
  volume = {28},
  number = {4},
  pages = {1916--1929},
  year = {2014},
  doi = {10.1137/120862612},
  eprint = {1407.2053},
  archivePrefix = {arXiv},
  url = {https://arxiv.org/abs/1407.2053}
}

@inproceedings{KLMNU16,
  author = {Kant{\'e}, Mamadou Moustapha and Limouzy, Vincent and Mary, Arnaud and Nourine, Lhouari and Uno, Takeaki},
  title = {A Polynomial Delay Algorithm for Enumerating Minimal Dominating Sets in Chordal Graphs},
  booktitle = {Graph-Theoretic Concepts in Computer Science ({WG} 2015)},
  series = {Lecture Notes in Computer Science},
  volume = {9224},
  pages = {138--153},
  publisher = {Springer},
  year = {2016},
  doi = {10.1007/978-3-662-53174-7_11},
  url = {https://link.springer.com/chapter/10.1007/978-3-662-53174-7_11}
}

@article{BDHPR20,
  author = {Bonamy, Marthe and Defrain, Oscar and Heinrich, Marc and Pilipczuk, Micha{\l} and Raymond, Jean-Florent},
  title = {Enumerating Minimal Dominating Sets in {$K_t$}-Free Graphs and Variants},
  journal = {ACM Transactions on Algorithms},
  volume = {16},
  number = {3},
  articleno = {39},
  pages = {39:1--39:23},
  numpages = {23},
  year = {2020},
  doi = {10.1145/3386686},
  eprint = {1810.00789},
  archivePrefix = {arXiv},
  url = {https://arxiv.org/abs/1810.00789}
}

@inproceedings{CDG25,
  author = {Castelo, Emanuel and Defrain, Oscar and Gomes, Guilherme C. M.},
  title = {Enumerating Minimal Dominating Sets and Variants in Chordal Bipartite Graphs},
  booktitle = {19th International Symposium on Algorithms and Data Structures ({WADS} 2025)},
  series = {Leibniz International Proceedings in Informatics (LIPIcs)},
  volume = {349},
  pages = {15:1--15:15},
  publisher = {Schloss Dagstuhl -- Leibniz-Zentrum f{\"u}r Informatik},
  year = {2025},
  doi = {10.4230/LIPIcs.WADS.2025.15},
  url = {https://drops.dagstuhl.de/entities/document/10.4230/LIPIcs.WADS.2025.15}
}

@article{KKMMO25,
  author = {Kobayashi, Yasuaki and Kurita, Kazuhiro and Mann, Kevin and Matsui, Yasuko and Ono, Hirotaka},
  title = {Enumerating Minimal Vertex Covers and Dominating Sets with Capacity and/or Connectivity Constraints},
  journal = {Algorithms},
  volume = {18},
  number = {2},
  pages = {112},
  year = {2025},
  doi = {10.3390/a18020112},
  eprint = {2308.16426},
  archivePrefix = {arXiv},
  url = {https://www.mdpi.com/1999-4893/18/2/112}
}

@article{BFS22,
  author = {Bl{\"a}sius, Thomas and Friedrich, Tobias and Schirneck, Martin},
  title = {The Complexity of Dependency Detection and Discovery in Relational Databases},
  journal = {Theoretical Computer Science},
  volume = {900},
  pages = {79--96},
  year = {2022},
  doi = {10.1016/j.tcs.2021.11.020},
  eprint = {2103.13331},
  archivePrefix = {arXiv},
  url = {https://arxiv.org/abs/2103.13331}
}

@article{IPZ01,
  author = {Impagliazzo, Russell and Paturi, Ramamohan and Zane, Francis},
  title = {Which Problems Have Strongly Exponential Complexity?},
  journal = {Journal of Computer and System Sciences},
  volume = {63},
  number = {4},
  pages = {512--530},
  year = {2001},
  doi = {10.1006/jcss.2001.1774},
  url = {https://cseweb.ucsd.edu/~paturi/myPapers/pubs/ImpagliazzoPaturiZane_2001_jcss.pdf}
}

@inproceedings{CGKMUW19,
  author = {Conte, Alessio and Grossi, Roberto and Kant{\'e}, Mamadou Moustapha and Marino, Andrea and Uno, Takeaki and Wasa, Kunihiro},
  title = {Listing Induced {Steiner} Subgraphs as a Compact Way to Discover {Steiner} Trees in Graphs},
  booktitle = {44th International Symposium on Mathematical Foundations of Computer Science ({MFCS} 2019)},
  series = {Leibniz International Proceedings in Informatics (LIPIcs)},
  volume = {138},
  pages = {73:1--73:14},
  publisher = {Schloss Dagstuhl -- Leibniz-Zentrum f{\"u}r Informatik},
  year = {2019},
  doi = {10.4230/LIPIcs.MFCS.2019.73},
  url = {https://drops.dagstuhl.de/entities/document/10.4230/LIPIcs.MFCS.2019.73}
}

@inproceedings{KK20,
  author = {Kurita, Kazuhiro and Kobayashi, Yasuaki},
  title = {Efficient Enumerations for Minimal Multicuts and Multiway Cuts},
  booktitle = {45th International Symposium on Mathematical Foundations of Computer Science ({MFCS} 2020)},
  series = {Leibniz International Proceedings in Informatics (LIPIcs)},
  volume = {170},
  pages = {60:1--60:14},
  publisher = {Schloss Dagstuhl -- Leibniz-Zentrum f{\"u}r Informatik},
  year = {2020},
  doi = {10.4230/LIPIcs.MFCS.2020.60},
  url = {https://drops.dagstuhl.de/entities/document/10.4230/LIPIcs.MFCS.2020.60}
}

@inproceedings{KKW22,
  author = {Kobayashi, Yasuaki and Kurita, Kazuhiro and Wasa, Kunihiro},
  title = {Linear-Delay Enumeration for Minimal {Steiner} Problems},
  booktitle = {Proceedings of the 41st ACM SIGMOD-SIGACT-SIGAI Symposium on Principles of Database Systems ({PODS} 2022)},
  pages = {301--313},
  publisher = {Association for Computing Machinery},
  year = {2022},
  doi = {10.1145/3517804.3524148},
  eprint = {2010.11462},
  archivePrefix = {arXiv},
  url = {https://arxiv.org/abs/2010.11462}
}

@article{BDM25,
  author = {Bergougnoux, Benjamin and Defrain, Oscar and {Mc Inerney}, Fionn},
  title = {Enumerating Minimal Solution Sets for Metric Graph Problems},
  journal = {Algorithmica},
  volume = {87},
  pages = {712--735},
  year = {2025},
  doi = {10.1007/s00453-025-01300-4},
  eprint = {2309.17419},
  archivePrefix = {arXiv},
  url = {https://arxiv.org/abs/2309.17419}
}

@inproceedings{KM26,
  author = {Kurita, Kazuhiro and Mann, Kevin},
  title = {On the Complexity of Hyperpath and Minimal Separator Enumeration in Directed Hypergraphs},
  booktitle = {{WALCOM}: Algorithms and Computation},
  series = {Lecture Notes in Computer Science},
  volume = {16444},
  pages = {560--575},
  publisher = {Springer},
  year = {2026},
  doi = {10.1007/978-981-95-7127-7_37},
  eprint = {2507.07528},
  archivePrefix = {arXiv},
  url = {https://link.springer.com/chapter/10.1007/978-981-95-7127-7_37}
}

@article{Kha95,
  author = {Khardon, Roni},
  title = {Translating between {Horn} Representations and Their Characteristic Models},
  journal = {Journal of Artificial Intelligence Research},
  volume = {3},
  pages = {349--372},
  year = {1995},
  doi = {10.1613/jair.183},
  url = {https://homes.luddy.indiana.edu/rkhardon/PUB/jair95-translating.pdf}
}

@inproceedings{EM02,
  author = {Eiter, Thomas and Makino, Kazuhisa},
  title = {On Computing All Abductive Explanations},
  booktitle = {Proceedings of the Eighteenth National Conference on Artificial Intelligence ({AAAI} 2002)},
  pages = {62--67},
  publisher = {AAAI Press},
  year = {2002},
  url = {https://cdn.aaai.org/AAAI/2002/AAAI02-010.pdf}
}

@inproceedings{dCM22,
  author = {de Colnet, Alexis and Marquis, Pierre},
  title = {On the Complexity of Enumerating Prime Implicants from {Decision-DNNF} Circuits},
  booktitle = {Proceedings of the Thirty-First International Joint Conference on Artificial Intelligence ({IJCAI} 2022)},
  pages = {2583--2590},
  publisher = {International Joint Conferences on Artificial Intelligence Organization},
  year = {2022},
  doi = {10.24963/ijcai.2022/358},
  url = {https://www.ijcai.org/proceedings/2022/358}
}

@article{DNV21,
  author = {Defrain, Oscar and Nourine, Lhouari and Vilmin, Simon},
  title = {Translating between the Representations of a Ranked Convex Geometry},
  journal = {Discrete Mathematics},
  volume = {344},
  number = {7},
  pages = {112399},
  year = {2021},
  doi = {10.1016/j.disc.2021.112399},
  eprint = {1907.09433},
  archivePrefix = {arXiv},
  url = {https://arxiv.org/abs/1907.09433}
}

@article{AGK11,
  author = {Allamigeon, Xavier and Gaubert, St{\'e}phane and Katz, Ricardo D.},
  title = {Tropical Polar Cones, Hypergraph Transversals, and Mean Payoff Games},
  journal = {Linear Algebra and its Applications},
  volume = {435},
  number = {7},
  pages = {1549--1574},
  year = {2011},
  doi = {10.1016/j.laa.2011.02.004},
  eprint = {1004.2778},
  archivePrefix = {arXiv},
  url = {https://arxiv.org/abs/1004.2778}
}

@article{EiterMG08,
  author       = {Thomas Eiter and
                  Kazuhisa Makino and
                  Georg Gottlob},
  title        = {Computational aspects of monotone dualization: {A} brief survey},
  journal      = {Discret. Appl. Math.},
  volume       = {156},
  number       = {11},
  pages        = {2035--2049},
  year         = {2008},
  url          = {https://doi.org/10.1016/j.dam.2007.04.017},
  doi          = {10.1016/J.DAM.2007.04.017},
  bibsource    = {dblp computer science bibliography, https://dblp.org}
}

@article{BuzzegaCKKP25,
  author       = {Giovanni Buzzega and
                  Alessio Conte and
                  Yasuaki Kobayashi and
                  Kazuhiro Kurita and
                  Giulia Punzi},
  title        = {The Complexity of Maximal Common Subsequence Enumeration},
  journal      = {Proc. {ACM} Manag. Data},
  volume       = {3},
  number       = {2},
  pages        = {115:1--115:19},
  year         = {2025},
  url          = {https://doi.org/10.1145/3725252},
  doi          = {10.1145/3725252},
  bibsource    = {dblp computer science bibliography, https://dblp.org}
}

@article{BDMN26,
  author = {Bonamy, Marthe and Defrain, Oscar and Micek, Piotr and Nourine, Lhouari},
  title = {Enumerating Minimal Dominating Sets in the (In)comparability Graphs of Bounded Dimension Posets},
  journal = {Discrete Mathematics},
  volume = {349},
  number = {4},
  pages = {114904},
  year = {2026},
  doi = {10.1016/j.disc.2025.114904},
  eprint = {2004.07214},
  archivePrefix = {arXiv},
  url = {https://arxiv.org/abs/2004.07214}
}

@inproceedings{LokshtanovS21,
  author       = {Daniel Lokshtanov and
                  Vaishali Surianarayanan},
  editor       = {Miko{\l}aj Boja\'{n}czyk and
                  Chandra Chekuri},
  title        = {Dominating Set in Weakly Closed Graphs is Fixed Parameter Tractable},
  booktitle    = {41st {IARCS} Annual Conference on Foundations of Software Technology
                  and Theoretical Computer Science, {FSTTCS} 2021, Virtual Conference,
                  December 15-17, 2021},
  series       = {LIPIcs},
  volume       = {213},
  pages        = {29:1--29:17},
  publisher    = {Schloss Dagstuhl - Leibniz-Zentrum f{\"{u}}r Informatik},
  year         = {2021},
  url          = {https://doi.org/10.4230/LIPIcs.FSTTCS.2021.29},
  doi          = {10.4230/LIPICS.FSTTCS.2021.29},
  bibsource    = {dblp computer science bibliography, https://dblp.org}
}

@article{ImpagliazzoP01,
  author       = {Russell Impagliazzo and
                  Ramamohan Paturi},
  title        = {On the Complexity of $k$-{SAT}},
  journal      = {J. Comput. Syst. Sci.},
  volume       = {62},
  number       = {2},
  pages        = {367--375},
  year         = {2001},
  url          = {https://doi.org/10.1006/jcss.2000.1727},
  doi          = {10.1006/JCSS.2000.1727},
  bibsource    = {dblp computer science bibliography, https://dblp.org}
}

\end{document}